\documentclass[11pt]{article}

\usepackage{amsfonts,amsmath,amsthm,times,fullpage,amssymb}

\usepackage[hidelinks]{hyperref}

\usepackage[dvipsnames]{xcolor}

\usepackage[noend]{algpseudocode}
\usepackage[nothing]{algorithm}  
\usepackage{caption}
 
\newtheorem{theorem}{Theorem} 
\newtheorem{lemma}{Lemma}
\newtheorem{definition}{Definition} 
 
\newtheorem{remark}{Remark}  
  
\newtheorem{proposition}{Proposition}

\newtheorem*{harmonic}{Harmonic Walks}
\newtheorem*{regeneration}{Regeneration}
\newtheorem*{gLLL}{General LLL}

\newtheorem*{cd}{Potential Causality Digraph}
\newtheorem*{causality}{Potential Causality}

\newcommand\address{{\sc Address}}

\newcommand{\Indep}{\mathrm{Ind}}
\newcommand{\List}{\mathrm{List}}  
\newcommand{\Span}{\mathcal{S}}

\newcommand{\pres}{\xi}
\newcommand{\res}{\xi}

\newcommand{\param}{\psi}

\begin{document}

\title{Focused Stochastic Local Search and the Lov\'{a}sz Local Lemma}

\author{
Dimitris Achlioptas
\thanks{Research supported by NSF grant CCF-1514128.}
\\ Department of Computer Science\\ University of California Santa Cruz
\\
\and 
Fotis Iliopoulos\thanks{ 
Research supported by NSF grant CCF-1514434.} 
 \\ Department of Electrical Engineering and Computer Science\\ 
 University of California Berkeley}

\date{\empty}

\maketitle

\begin{abstract}
We develop tools for analyzing focused stochastic local search algorithms. These are algorithms which search a state space probabilistically by repeatedly selecting a constraint that is violated in the current state and moving to a random nearby state which, hopefully, addresses the violation without introducing many new ones. A large class of such algorithms arise from the algorithmization of the Lov\'{a}sz Local Lemma, a non-constructive tool for proving the existence of satisfying states. Here we give tools that provide a unified analysis of such algorithms and of many more, expressing them as instances of a general framework. 
\end{abstract}

\newpage
\section{Introduction}\label{sec:intro}

Let $\Omega$ be a large finite set of objects and let $F = \{f_1, f_2, \ldots, f_m \}$ be a collection of subsets of $\Omega$. We will refer to each $f_i \in F$ as a \emph{flaw} to express that its elements have negative feature $i \in [m]$. For example, given a CNF formula on $n$ variables with clauses $c_1,c_2,\ldots,c_m$, we can define for each clause $c_i$ the flaw (subcube) $f_i \subseteq \{0,1\}^n$ whose elements violate $c_i$. Following linguistic rather than mathematical convention we say that $f$ is present in $\sigma$ if $f \ni \sigma$ and that $\sigma \in \Omega$ is \emph{flawless} (perfect) if no flaw is present in $\sigma$. 

Our goal is to develop tools for analyzing stochastic local search algorithms for finding perfect objects. Importantly, our analysis \emph{will not} assume that $\Omega$ contains perfect objects, but rather will establish their existence by proving that an algorithm converges (quickly) to one. The general idea in stochastic local search is that $\Omega$ is equipped with a neighborhood structure and that the search starts at some element (state) of $\Omega$ and moves from state to state along the neighborhood structure. 
\emph{Focused} local search corresponds to the case where each state change can be attributed to an effort to rid the state of some specific present flaw. 
    
Concretely, for each $\sigma \in \Omega$, let $U(\sigma) = \{ f \in F: \sigma \in f \}$, i.e., $U(\sigma)$ is the set of flaws present in $\sigma$. For every flaw $f_i \in U(\sigma) $,  let $A(i,\sigma) \ne \{ \sigma \}$ be a non-empty subset of $\Omega$. We call the elements of $A(i,\sigma)$ \emph{actions} and we consider the multi-digraph $D$ on $\Omega$ which has an arc $\sigma \xrightarrow{i} \tau $ for every $\tau \in A(i,\sigma)  $. We will consider walks on $D$ which start at a state $\sigma_1$, selected according to some probability distribution $\theta$, and which at each non-sink vertex $\sigma$ first select a flaw $f_i \ni \sigma$, as a function of the trajectory so far, and then select a next state $\tau \in A(i,\sigma) $ with probability $\rho_i(\sigma, \tau)$. Whenever a flaw $f_i \ni \sigma$ is selected we will say flaw $f_i$ was addressed (which will not necessarily mean that $f_i$ was eliminated, i.e., potentially $A(i,\sigma) \cap f_i \ne \emptyset$). 

A large class of algorithms for dealing with the setting at this level of generality arise by algorithmizations of the Lov\'{a}sz Local Lemma~(LLL). This is a non-constructive tool for proving the \emph{existence} of flawless objects by introducing a probability measure $\mu$ on $\Omega$, along the lines of the Probabilistic Method
(throughout we assume that products devoid of factors evaluate to 1, i.e., $\prod_{x \in\emptyset} g(x) = 1$ for any $g$).

\begin{gLLL}
Let $\mathcal{A} = \{A_1, A_2,\ldots,A_m\}$ be a set of $m$ events. For each $i \in [m]$, let $D(i) \subseteq	 [m] \setminus \{i\}$ be such that $\mu(A_i \mid \cap_{j \in S} \overline{A_j}) = \mu(A_i)$ for every $S \subseteq [m] \setminus (D(i) \cup \{i\})$. If there exist positive real numbers $\{\psi_i\}_{i=1}^m$ such that for all $i \in [m]$,
\begin{equation}\label{eq:LLL}
\frac{\mu(A_i)}{\psi_i}  \sum_{ S \subseteq \{i \} \cup D(i) }  \prod_{j \in S} \psi_j \le 1 \enspace , 
\end{equation}
then the probability that none of the events in $\mathcal{A}$ occurs is at least $\prod_{i=1}^m 1/(1+\psi_i) > 0$. 
\end{gLLL}

\begin{remark}
Condition~\eqref{eq:LLL} above is equivalent to the more well-known form $\mu(A_i) \le x_i \prod_{j \in D(i)} (1-x_j)$, 
where $x_i = \psi_i/(1+\psi_i)$. As we will see, formulation~\eqref{eq:LLL} facilitates refinements. 
\end{remark}
Erd\H{os} and Spencer~\cite{LopsTrav} noted that independence in the LLL can be replaced by \emph{negative correlation}, yielding the stronger Lopsided LLL. The  difference is that each set $D(i)$ is replaced by a set $L(i) \subseteq [m] \setminus \{ i \}$ such that $\mu(A_i \mid \cap_{j \in S} \overline{A_j}) \le \mu(A_i)$ for every $S \subseteq [m] \setminus (L(i) \cup \{i\})$, i.e., ``=" is replaced by ``$\le$".\medskip

In a landmark work~\cite{MT}, Moser and Tardos made the general LLL constructive for \emph{product} measures over explicitly presented variables. Specifically, in the \emph{variable setting} of~\cite{MT}, each event $A_i$ is determined by a set of variables $\mathrm{vbl}(A_i)$ so that $ j \in D(i)$ iff $\mathrm{vbl}(A_i) \cap \mathrm{vbl}(A_j) \ne \emptyset$. Moser and Tardos proved that if~\eqref{eq:LLL} holds, then repeatedly selecting \emph{any} occurring event $A_i$ (flaw present) and resampling every variable in $\mathrm{vbl}(A_i)$ independently of all others, leads to a flawless object after a linear expected  number of resamplings. Pegden~\cite{PegdenIndepen} extended the result of~\cite{MT} to the cluster expansion criterion of Bissacott et al.~\cite{bissacot2011improvement}, and  Kolipaka and Szegedy~\cite{szege_meet} extended it to Shearer's criterion (the most general LLL criterion). Beyond the variable setting, Harris and Srinivasan in~\cite{SrinivasanPerm} algorithmized the general LLL for the uniform measure on permutations while, very recently, Harvey and Vondr\'{a}k~\cite{HV} algorithmized the Lopsided LLL up to Shearer's criterion assuming efficiently implementable \emph{resampling oracles}. Resampling oracles, introduced in~\cite{HV}, elegantly capture a common constraint in all prior algorithmizations of LLL, namely that state transitions must be ``compatible" with the measure $\mu$. Below we give the part of the definition of resampling oracles that exactly expresses this notion of compatibility, which we dub (measure) \emph{regeneration}.
\begin{regeneration}\emph{(\cite{HV})}[Resampling Oracles]
\label{def:loc_regen}
Say that $(D,\rho)$ \emph{regenerate $\mu$ at flaw $f_i$} if for every $\tau \in \Omega$,
\begin{equation}\label{eq:oracle}
\frac{1}{\mu(f_i)}\sum_{\sigma \in f_i} \mu(\sigma) \rho_{i}(\sigma,\tau) = \mu(\tau) \enspace .
\end{equation}
\end{regeneration}

Observe that the l.h.s.\ of~\eqref{eq:oracle} is the probability of reaching $\tau$ after first sampling a state $\sigma \in f_i$ according to $\mu$ and then addressing $f_i$ at $\sigma$ according to $(D,\rho)$. The requirement that this probability equals $\mu(\tau)$ for every $\tau \in \Omega$ means that in every state $\sigma \in f_i$ the distribution of actions for addressing $f_i$ must be such that it removes the conditional $f_i \ni \sigma$. A trivial way to satisfy this (very stringent) requirement is to sample a new state $\sigma'$ according to $\mu$ in each step (assuming $\mu$ is efficiently sampleable). Doing this, though, removes any sense of progress, as the set of flaws present in $\sigma'$ are completely unrelated to those in $\sigma$. Instead, we would like to achieve~\eqref{eq:oracle} while limiting the set of flaws that may be present in $\sigma'$ that were not present in $\sigma$. For example, note that in the variable setting resampling every variable in $\mathrm{vbl}(f_i)$ independently satisfies~\eqref{eq:oracle}, while only having the potential to introduce flaws that share at least one variable with $f_i$. It is, thus, natural to consider the following ``projection" of the action digraph introduced in the flaws/actions framework of~\cite{AI}.

\begin{causality}
For an arc $\sigma \xrightarrow{i}  \tau$ in $D$ and a flaw $f_j$ present in $\tau$ we say that $f_i$ causes $f_j$ if $f_i = f_j$ or $f_j \not\ni \sigma$. If $D$ contains \emph{any} arc in which $f_i$ causes $f_j$ we say that $f_i$ \emph{potentially causes} $f_j$.
\end{causality}

\begin{cd}
The digraph $C=C(\Omega,F,D)$ on $[m]$ where $i \rightarrow j$ iff $f_i$ potentially causes $f_j$ is called the potential causality digraph. The \emph{neighborhood} of a flaw $f_i$  is $\Gamma(i) =\{j : i \to  j \text{  exists in $C$}\}$.
\end{cd}
In the interest of brevity we will call $C$ the causality digraph, instead of the potential causality digraph. It is important to note that $C$ contains an arc $i \to j$ if there exists \emph{even one} state transition aimed at addressing $f_i$ that causes $f_j$ to appear in the new state. 

As mentioned, very recently, Harvey and Vondr\'{a}k~\cite{HV} made the Lopsided LLL algorithmic, given resampling oracles for $\mu$. Their result actually makes no reference to the lopsidependency condition, which they prove is implied by the existence of resampling oracles, and can be stated as follows. 

\begin{theorem}[Harvey-Vondr\'{a}k~\cite{HV}]\label{thm:HV}
Let $\Omega, F,\mu, D, \rho$ be such that $(D,\rho)$ regenerate $\mu$ at flaw $f_i$ for every $i \in [m]$. If $\theta = \mu$ and  there exist positive real numbers $\{\psi_i\}_{i=1}^m$ such that for every $i \in [m]$,
\begin{equation}\label{eq:HVLLL}
\frac{\mu(f_i)}{\psi_i}  \sum_{ S \subseteq \Gamma(i)}   \prod_{j \in S} \psi_j  < 1 \enspace , 
\end{equation}
then a perfect object can be found after polynomially many steps on $D$.
\end{theorem}

In fact, in~\cite{HV} it was shown that the conclusion of Theorem~\ref{thm:HV} holds also if~\eqref{eq:HVLLL} is replaced by Shearer's condition~\cite{Shearer}. Thus, the work of Harvey and Vondr\'{a}k~\cite{HV} marks the end of the road for the derivation and analysis of focused stochastic local search algorithms by algorithmizations of the Lov\'{a}sz Local Lemma.

In this work we extend the flaws/actions algorithmic framework of~\cite{AI} to arbitrary measures and action graphs and connect it to the Lov\'{a}sz Local Lemma. The result is a theorem that subsumes both Theorem~\ref{thm:HV} and all results of~\cite{AI}, establishing a method for designing and analyzing focused stochastic local search algorithms that goes far beyond algorithmizing the LLL.

Concretely, in~\cite{AI} we introduced and analyzed focused local search algorithms where:
\begin{itemize}
\item
$D$ is \emph{atomic}, i.e., for every $\tau \in \Omega$ and every $i\in [m]$ there exists at most one arc $\sigma \xrightarrow{i} \tau$.
\item
$\mu$ is the \emph{uniform} measure on $\Omega$, i.e., $\mu(\cdot) = |\Omega|^{-1}$.
\item
$\rho$ assigns \emph{equal} probability to every action in $A(i,\sigma)$, for every $f_i \in U(\sigma)$, at every flawed $\sigma \in \Omega$.  
\end{itemize}
Here we generalize to {\bf arbitrary} $D,\rho,\mu$, allowing one to trade the sophistication of the measure $\mu$ against the sparsity of the causality graph (while removing the need to sample from $\mu$, or to regenerate $\mu$). Moreover, for the special case of uniform $\mu$, we improve the condition of~\cite{AI} for convergence. We state our results formally in the next section. \medskip

We also make a conceptual contribution by identifying for each measure $\mu$ certain pairs $(D,\rho)$ as special.
\begin{harmonic}
$(D,\rho,\mu)$ are \emph{harmonic} if for every $i \in [m]$ and every transition $(\sigma,\tau ) \in   f_i \times A(i,\sigma)$,
\begin{equation}\label{eq:harmonic_rho_def} 
\rho_i(\sigma,\tau) =  \frac{\mu(\tau)}{ \sum_{\sigma' \in A(i,\sigma) } \mu(\sigma')} \propto \mu(\tau) \enspace .
\end{equation}
\end{harmonic}
In words, when $(D,\rho,\mu)$ are harmonic $\rho_i$ assigns to each state in $A(i,\sigma)$ probability proportional to its probability under $\mu$. It is easy to see that $(D,\rho,\mu)$ are harmonic both in the algorithm of Moser and Tardos~\cite{MT} for the variable setting and in the algorithm of Harris and Srinivasan~\cite{SrinivasanPerm} for the uniform measure on permutations. There are two reasons why harmonic $(D,\rho,\mu)$ combinations are interesting.

\subsection{Resampling Oracles via Atomic Actions}

If we start with the Probabilistic Method setup, i.e., $\Omega, F$, and $\mu$, then to get a constructive result by LLL algorithmization we must design $(D,\rho)$ that regenerate $\mu$ at every flaw $f_i \in F$ (note that Theorem~\ref{thm:HV} \emph{assumes} such $(D,\rho)$ as input). While, in general, this can be a daunting task, we show that if we restrict our attention to $D$ that are atomic matters are dramatically simplified:
\begin{itemize}
\item
$(D,\rho, \mu)$ \emph{must} be harmonic, yielding a \emph{local} characterization of $(D,\rho)$ at every state. \hfill(Theorem~\ref{atomic_oracles})
\item 
The probability of every sequence of states is \emph{characterized} by the flaws addressed. \hfill(Theorem~\ref{tight})
\item
The initial state can be \emph{arbitrary}, i.e., we can have $\theta \neq \mu$. \hfill(Theorems~\ref{asymmetric} --
\ref{lem:master})
\end{itemize}

We note that \emph{all} previous LLL algorithmizations, including~\cite{HV}, require $\theta = \mu$. We remove this at a mere cost of adding $O(\log|\Omega|)$ to the running time. This  is beneficial in settings where sampling  from $\mu$ (a global property) is hard, but nonetheless we can  regenerate $\mu$ at every flaw (a local property).
\smallskip

Atomicity may initially seem artificial and/or restrictive. In reality, it is a very natural way to promote search space exploration, as it is equivalent to the following: $A(i,\sigma) \cap A(i,\sigma') = \emptyset$ for every $\sigma \neq \sigma' \in f_i$. Moreover, in most settings atomicity can be achieved in a straightforward manner. For example, in the variable setting the following two conditions combined imply atomicity:
\begin{enumerate}
\item\label{satlike}
Each flaw forbids exactly \emph{one} joint value assignment to its underlying variables, i.e., is a subcube. 
\item\label{nospooky}
Each state transition modifies \emph{only} the variables of the flaw it addresses. 
\end{enumerate} 

Condition~\ref{satlike} expresses a primarily syntactic requirement: compound constraints must be broken down to constituent parts akin to satisfiability constraints. In most settings, not only is such a breakdown straightforward, but is also advantageous, as it affords a more refined accounting of conflict between constraints.

Condition~\ref{nospooky} on the other hand reflects ``focusing", i.e., that every state transformation should be the result of attempting to eradicate some specific flaw $f_i$ without interfering with variables not in $\mathrm{vbl}(f_i)$. 

\subsection{Beyond Measure Regeneration}

Designing $(D,\rho)$ to regenerate $\mu$ at every flaw can be highly restrictive. This is commonly demonstrated by LLL's inability to establish that a graph with maximum degree $\Delta$ can be colored with $ q = \Delta + 1$ colors, one of the oldest and most vexing concerns about the LLL (see the survey of Szegedy~\cite{mario_survey}). This is because the regeneration of the uniform measure implies that to recolor a vertex $v$ we must select uniformly among all colors, rather than colors not appearing in $v$'s neighborhood, inducing a requirement of $q > \mathrm{e}\Delta$ colors. 

In~\cite{AI} we introduced the flaws/actions framework to initiate the study of local search algorithms whose actions can depend arbitrarily on the state. In the aforementioned example this could mean choosing only among colors not appearing in $v$'s neighborhood, so that as soon as $q \ge \Delta+1$, the causality digraph becomes empty and rapid termination follows trivially. In~\cite{AI}, we required $D$ to be atomic and actions to be chosen uniformly, a setting that in our current framework can be seen as the special case where $\mu$ is uniform, $D$ is atomic, and $(D,\rho,\mu)$ are harmonic. Here we consider general harmonic $(D,\rho,\mu)$, i.e., $D$ need not be atomic and $\mu$ can be arbitrary, while $(D,\rho)$ need not regenerate $\mu$. Rather, $\mu$ is only used as a gauge of progress and deviation from regeneration is traded-off against sparsity of the causality graph.

The reason we focus on harmonic $(D,\rho,\mu)$ is that, as we will see, (i) they are optimal with respect to the aforementioned trade-off,  and (ii) for every $D$ and $\mu$ there exists $\rho$ such that $(D,\rho,\mu)$ are harmonic.

\section{Statement of Results}\label{sec:results}

\begin{definition}\label{def:nameless}
For $i \in [m]$ and  $\tau \in \Omega$, let $b_i^{\tau} = |\{\sigma \in f_i: \tau \in A(i,\sigma)\}|$. For $i \in [m]$, let  $b_i = \max_{\tau \in \Omega } b_i^{\tau}$. If $b_i = 1$ for all $i \in [m]$, then we say that $D$ is \emph{atomic}. (Note that $b_i>0$ since $A(i,\sigma) \neq \emptyset$ for $\sigma \in f_i$.)
\end{definition}

\subsection{Setup}\label{sec:setup}
We establish general conditions under which focused stochastic local search algorithms find flawless objects quickly. Recall that any such algorithm performs a random walk on a multi-digraph $D$ which (i) starts at a state $\sigma_1 \in \Omega$ selected according to a distribution $\theta$, and which (ii) at each flawed state $\sigma$ first selects some $f_i \ni \sigma$ to address and then selects $\tau \in A(i,\sigma)$ as the next state, where each $\tau \in A(i,\sigma)$ is selected with probability $\rho_i(\sigma,\tau)$. As one may expect the flaw-choice mechanism does have a bearing on the running time of such algorithms and we discuss this point in Section~\ref{sec:fc}. Our results focus on conditions for rapid termination that do not require sophisticated flaw choice (but can be used in conjunction which such choice). 

To establish a walk's capacity to rid the state of flaws we introduce a measure $\mu$ on $\Omega$. Without loss of generality, and to avoid certain trivialities, we assume that $\mu(\sigma)>0$ for all $\sigma \in \Omega$.  The choice of $\mu$ is entirely ours and can be trivial, i.e., $\mu(\cdot) = |\Omega|^{-1}$. Typically, $\mu$ assigns only exponentially small probability to flawless objects, yet allows us to prove that the walk reaches such an object in polynomial expected time. Its role is to define a ``charge" $\gamma_i = \gamma_i(D,\theta,\rho,\mu)$ for each flaw $f_i \in F$, ideally as small as possible.

\subsection{Definition of Flaw Charges}

\noindent {\bf Regenerative case.}
If $(D,\rho)$ regenerate $\mu$ at every $f_i \in F$ and either $\theta = \mu$ or $D$ is atomic, then 
\[
\gamma_i = \mu(f_i) \enspace .
\]

\noindent {\bf General case.}
Otherwise, $\gamma_i  = b_i  \max_{\sigma \in f_i} \lambda_i^{\sigma}$, where
\begin{equation}\label{eq:lambda_def}
\lambda_i^{\sigma} = \max_{\tau \in A(i,\sigma)} \left\{\rho_i(\sigma,\tau) \, \frac{\mu(\sigma)}{\mu(\tau)}\right\} \enspace .
\end{equation}

We will discuss several aspects of the definition of $\gamma_i$ in Section~\ref{charging}. The main point is that the notion of charge allows us to state our results \emph{without} having to distinguish between the regenerative and the general case, by simply substituting the appropriate charge. This also serves to highlight that the standard (Probabilistic Method) formulation of the LLL (and its algorithmizations) is, in fact, only a facet of a far more general picture, for which our results provide the first analytical tools. In the regenerative case, since $\gamma_i = \mu(f_i)$, our conditions will parallel those of the LLL (and its algorithmizations). In the general (non-regenerative) case, we will have $\gamma_i \ge \mu(f_i)$ always, but a potentially far sparser causality graph. To gain some first intuition for $\gamma_i$ as a notion of congestion in the general case observe that if $\mu$ is uniform and $D$ is atomic, then $\gamma_i$ is simply the greatest transition probability $\rho_i(\sigma,\tau)$ on any arc originating in $f_i$. In general, it is the ergodic flow from $\sigma$ to $\tau$ divided by the capacity, $\mu(\tau)$, of $\tau$ (and scaled by $b_i$). 

To state our results we need a last definition regarding the distribution $\theta$ of the initial state. 

\begin{definition}
The \emph{span} of a probability distribution $\theta: \Omega \rightarrow [0,1]$, denoted by $\Span(\theta)$, is the set of flaws that may be present in a state selected according to $\theta$, i.e., $ \Span(\theta) = \bigcup_{ \sigma \in \Omega: \theta(\sigma) > 0 } U(\sigma)$.
\end{definition}
 
\subsection{A Simple Markov Chain}
 
Our first result concerns the simplest case where, after choosing an arbitrary permutation $\pi$ of the flaws, the algorithm in each flawed state $\sigma$ simply addresses the greatest flaw present in $\sigma$ according to $\pi$. Observe that substituting $\gamma_i = \mu(f_i)$ for the regenerative case to Theorem~\ref{asymmetric} recovers 
the condition of Theorem~\ref{thm:HV}.
\begin{theorem}\label{asymmetric}
If there exist positive real numbers $\{\psi_i\}$ such that for every $i \in [m]$,
\begin{align}
\zeta_i := 
\frac{\gamma_i}{\psi_i}
 \sum_{ S \subseteq   \Gamma(i) }  \prod_{j \in S} \psi_j   < 1\enspace , \label{eq:mc2}
\end{align}
then for every $\pi$ the walk reaches a sink within $(T_0+s)/\delta$ steps with probability at least $1-2^{-s}$, where $\delta = 1 - \max_{i \in [m]} \zeta_i 
> 0$, and
\[
T_0 	 = 
\log_2 \left( \max_{ \sigma \in \Omega } \frac{\theta(\sigma)  }{ \mu(\sigma) }\right)  + \log_2 \left( \sum_{S \subseteq \Span(\theta) } \prod_{j \in S} \psi_j   \right)\enspace .
\]
\end{theorem}

Theorem~\ref{asymmetric} has two features worth discussing, both directs consequences of the generality of our framework, i.e., of abandoning the Probabilistic Method viewpoint and measure regeneration.\medskip

\noindent {\bf Arbitrary initial state.} Since $\theta$ can be arbitrary in the general case, any foothold on $\Omega$ suffices to apply the theorem, without needing to sample from $\Omega$ according to some measure. This can also be interesting when we can not sample from $\mu$, but can regenerate $\mu$ at every flaw on an atomic $D$, i.e., the second case of the regenerative setting. Note also that $T_0$ captures the trade-off between the fact that when $\theta = \mu$ the first term in $T_0$ vanishes, but the second term grows to reflect the uncertainty of the set of flaws present in $\sigma_1$.\smallskip

\noindent {\bf Arbitrary number of flaws.} The running time depends only on the span $|\Span(\theta)|$, not the total number of flaws $|F|$. This has an implication  analogous to the result of Hauepler, Saha, and Srinivasan~\cite{haeupler_lll} on core events: even when $|F|$ is very large, e.g., super-polynomial in the problem's encoding length, we can still get an efficient algorithm if, for example,  we can find a state $\sigma_1$ such that $|U(\sigma_1)|$ is small, e.g., by proving that in every state only polynomially many flaws may be present, or $\theta$ such that $|\Span(\theta)|$ is small.\medskip

\subsection{A Non-Markovian Algorithm}

Our next results concerns the common setting where the subgraph induced by the neighborhood of each flaw in the causality graph  contains several arcs.  We improve Theorem~\ref{asymmetric} in such settings by employing a \emph{recursive} algorithm. The flaw addressed in each step thus depends on the \emph{entire} trajectory up that point not just the current state, i.e., the walk is non-Markovian. It is for this reason that we required a non-empty set of actions for every flaw present in a state, and why the definition of the causality digraph does not involve flaw choice. The improvement is that rather than summing over all subsets of $\Gamma(i)$ as in~\eqref{eq:mc2}, we now only sum over \emph{independent} such subsets, where $f_i,f_j$ are dependent if $f_i \rightarrow f_j$ and $f_j \rightarrow f_i$. This improvement is similar to the cluster expansion improvement of Bissacot et al.~\cite{bissacot2011improvement} of the general LLL. As a matter of fact, Theorem~\ref{olala} implies the algorithmic aspects of~\cite{bissacot2011improvement} (see~\cite{PegdenIndepen} and ~\cite{HV}). 

The use of a recursive algorithm affords an additional advantage, as it enables ``responsibility shifting" between flaws. Specifically, for a fixed action digraph $D$ with causality digraph $C$, the recursive algorithm (and Theorem~\ref{olala}), take as input \emph{any} digraph $R \supseteq C$, i.e., allow for arcs to be added to the causality digraph. The reason for this as follows. While adding, say, arcs $f_i\to f_j$ and $f_j \to f_i$ may make the sums corresponding to $f_i$ and $f_j$ greater, if $f_k$ is such that $\{f_i,f_j\} \subseteq \Gamma(k)$, then its sum may become smaller, as $f_i,f_j$ are now dependent. As a result, such arc addition may enable a sufficient condition for rapid convergence to a perfect object, e.g., in our application on Acyclic Edge Coloring in Section~\ref{AECARA}. An analogous counter-intuitive phenomenon is also true in the improvement of Bissacot et al.~\cite{bissacot2011improvement} where denser dependency graphs may result to a better analysis.

Below, for $S \subseteq F$, we let $I_{\pi}(S)$ denote the greatest element of $S$ according to $\pi$. For any fixed ordering $\pi$ of $F$ the recursive walk is the non-Markovian random walk on $\Omega$ that occurs by invoking procedure {\sc Eliminate}. Observe that if in line~\ref{a_key_diff} we do not intersect $U(\sigma)$ with $\Gamma_R(f_i)$ the recursion is trivialized and we recover the simple walk of Theorem~\ref{asymmetric}.

\begin{algorithm}[h]\caption*{{\bf Recursive Walk}}
\begin{algorithmic}[1]\label{Recursive}
\Procedure{Eliminate}{}
\State $\sigma \leftarrow \theta(\cdot)$ \Comment{Sample $\sigma$ from $\theta$}
\While {$U(\sigma) \neq \emptyset$}	
	\State {\sc Address} ($I_{\pi}(U(\sigma)),\sigma$) 
\EndWhile	
\State	\Return $\sigma$ 
\EndProcedure{}{}
\Procedure{Address}{$i,\sigma$}
\State $\sigma \leftarrow$ $\tau \in A(i,\sigma)$ with probability $\rho_i(\sigma,\tau)$ 
\While {$B = U(\sigma) \cap \Gamma_{R}(f_i) \neq \emptyset$}   \label{a_key_diff}	\Comment{Note $\,\cap \Gamma_{R}(f_i)$} \label{code:critical}
	\State {\sc{Address}}($I_{\pi}(B),\sigma$) 				\label{code:act}
\EndWhile  
\EndProcedure  
\end{algorithmic}
\end{algorithm}

\begin{definition}\label{defn:G}
Given a digraph $R$ on $F$ let $G=G(R)=(F,E)$ be the \emph{undirected} graph where $\{f,g\} \in E$ iff both $f \rightarrow g$ and $g \rightarrow f$ exist in $R$. For $S \subseteq F$, let $\Indep(S) = \{S' \subseteq S : \text{$S'$ is an independent set in $G$}\}$.
\end{definition}

\begin{theorem}\label{olala}
Let $R \supseteq C$ be arbitrary. If there exist positive real numbers $\{\psi_i\}$ such that for every $i \in [m]$,
\begin{align}
\zeta_i := \frac{\gamma_i}{\psi_i} \sum_{S \in \Indep(\Gamma_{R}(i))}  \prod_{j \in S} \psi_j < 1 \enspace , \label{eq:mc3}
\end{align}
then for every $\pi$ the recursive walk reaches a sink within $(T_0+s)/\delta$ steps with probability at least $1-2^{-s}$, where $\delta = 1 - \max_{i \in [m]} \zeta_i > 0$, and
\[
T_0 	 = 
\log_2 \left( \max_{ \sigma \in \Omega } \frac{\theta(\sigma)  }{ \mu(\sigma) }\right) + \log_2 \left( \sum_{S \subseteq \Indep \left( \Span(\theta) \right)  } \prod_{j \in S} \psi_j   \right)
 \enspace  .
\]
\end{theorem}

\begin{remark}
Theorem~\ref{olala} strictly improves Theorem~\ref{asymmetric} since by taking $R=C$ (i) the summation in~\eqref{eq:mc3} is only over the subsets of $\Gamma_R(f)$ that are independent in $G$, instead of  all subsets of $\Gamma_R(f)$ as in~\eqref{eq:mc2}, and (ii) similarly for $T_0$, the summation is only over the independent subsets of $\Span(\theta)$, rather than all subsets of $\Span(\theta)$.
\end{remark}

\begin{remark}
Theorem~\ref{olala} can be strengthened by introducing for each flaw $f \in F$ a permutation $\pi_f$ of $\Gamma_{R(f)}$ and replacing $\pi$ with $\pi_f$  in line~\ref{code:act} the of Recursive Walk. With this change in~\eqref{eq:mc3} it suffices to sum only over $S \subseteq \Gamma_{R}(f)$ satisfying the following: if the subgraph of $R$ induced by $S$ contains an arc $g \to h$, then $\pi_{f}(g) \ge \pi_f(h)$. As such a subgraph can not contain both $g \to h$ and $h \to g$ we see that $S \in \Indep(\Gamma_R(f))$.
\end{remark} 

\subsection{A General Theorem}

Theorems~\ref{asymmetric} and~\ref{olala} are instantiations of a general theorem we develop for establishing the success of focused local search algorithms by local considerations. Before presenting the theorem itself, we first briefly discuss its derivation, as that helps motivate and digest the theorem's form.

To bound the probability of not reaching a sink within $t$ steps we partition the set of all $t$-trajectories into equivalence classes, bound the total probability of each class, and sum the bounds for the different classes. The partition is according to the $t$-sequence of flaws addressed, which acts as a \emph{statistic} of the state distribution. Formally, for a trajectory $\Sigma = \sigma_1 \xrightarrow{w_1}   \sigma_2 \xrightarrow{w_2}  \cdots$ we let $W(\Sigma) = w_1,w_2\cdots$ denote its \emph{witness} sequence, i.e., the sequence of flaws addressed along $\Sigma$. We let $W_t(\Sigma) = \perp$ if $\Sigma$ has fewer than $t$ steps, otherwise we let $W_t(\Sigma)$ be the $t$-prefix of $W(\Sigma)$. Slightly abusing notation we let $W_t = W_t(\Sigma)$ be the random variable when $\Sigma$ is the trajectory of the walk, i.e., selected according to $(D,\rho,\theta)$ and the flaw-choice mechanism. If $\mathcal{W}_{t} = \mathcal{W}_{t}(\mathcal{A})$ denotes the range of $W_t$ for an algorithm $\mathcal{A}$, then the probability that $\mathcal{A}$ takes $t$ or more steps, trivially, is $\sum_{W \in \mathcal{W}_t}  \Pr[W_t = W]$.
 
Key to our analysis is deriving upper bounds for $\Pr[W_t =W]$ that factorize over the elements of $W$. Specifically, for an arbitrary sequence of flaws $A = a_1,\ldots,a_t$, let us denote by $[i]$ the index $j \in [m]$ such that $a_i = f_j$. Lemma~\ref{lemma:regeneration} holds for both the regenerative and the general case, with the corresponding $\gamma_i$. Moreover, we will see that it can be tight, up to the prefactor $\xi$.
\begin{lemma}\label{lemma:regeneration}
Let $\xi = \xi(\theta,\mu) = \max_{\sigma \in \Omega} \{\theta(\sigma)/\mu(\sigma)\}$. For every sequence of flaws $A = a_1,\ldots,a_t$,
\[
\Pr[W_t = A] \le \xi \prod_{i = 1}^{t} \gamma_{[i]} \enspace .
\]
\end{lemma}

The product form of the bound in Lemma~\ref{lemma:regeneration} allows us to combine it with different collections of forests, each collection expressing an upper bound for (superset of) the set of all possible witness sequences $\mathcal{W}_t$. The formulation of the supersets as forests enables the combinatorial enumeration of their elements (flaw sequences) which, combined with Lemma~\ref{lemma:regeneration}, yields Theorem~\ref{lem:master} below. While $\mathcal{W}_t$ depends on flaw-choice, the main and common feature of all bounds (forests) is the enforcement of the following idea: while the very first occurrence of a flaw $f_j$ in a witness sequence $W$  may be attributed to $f_j \ni \sigma_1$, every subsequent occurrence of $f_j$ must be preceded by a distinct earlier occurrence of a flaw $f_i$ that can  ``assume responsibility" for $f_j$, e.g., a flaw $f_i$ that potentially causes $f_j$. In this way, the set $\mathcal{W}_t$ is bounded \emph{syntactically} by differently sophisticated considerations of flaw-choice and responsibility. Specifically, Definition~\ref{def:traceable} below (i) imposes a modicum of control over flaw-choice, while (ii) generalizing the subsets of flaws for which a flaw $f$ may be responsible from subsets of $\Gamma(f)$ to arbitrary subsets of flaws, thus enabling responsibility shifting.

\begin{definition}\label{def:traceable}
Given $(D,\rho,\theta)$, a flaw-choice mechanism is \emph{traceable} if there exist sets $\mathrm{Roots}(\theta) \subseteq 2^{F}$ and $\mathrm{List}(f_1) \subseteq 2^{F}, \ldots, \mathrm{List}(f_m) \subseteq 2^{F}$ such that for every $t \ge 1$,  the set of all possible witness sequences $\mathcal{W}_t$ can be injected into unordered rooted forests with $t$ vertices that have the following properties:
\begin{enumerate}
\item 
Each vertex of the forest is labeled by a flaw $f_i\in F$.\label{cond:label}
\item 
The flaws labeling the roots of the forest are distinct and form an element of $\mathrm{Roots}(\theta)$.\label{cond:roots}
\item 
The flaws labeling the children of each vertex are distinct.\label{cond:distinct}
\item
If a vertex is labelled by flaw $f_i$, then the labels of its children form an element of $\mathrm{List}(f_i)$.
\label{cond:list}
\end{enumerate}
\end{definition}

To recover the witness sequence from a forest, thus demonstrating the injection of $\mathcal{W}_t$, we make use of the specificity of the mechanism for selecting which flaw to address at each step. For example, the forests that correspond to the algorithm of Theorem~\ref{olala} are ``recursion forests", having one node for each recursive call of \address, labelled by the flaw that is the call's argument. To recover the sequence of addressed flaws, we order the trees in the forest and the progeny of each vertex using knowledge of $\pi$ and then traverse each tree in the recursive forest in postorder. We explain why the algorithms of Theorems~\ref{asymmetric} and~\ref{olala} are traceable in Appendix~\ref{sec:forests}, where we describe the set of witness forests that correspond to each theorem.

Theorem~\ref{lem:master} below implies both Theorem~\ref{asymmetric} and Theorem~\ref{olala}. While those two theorems do not care about the flaw ordering $\pi$,  Theorem~\ref{lem:master} also captures the  ``LeftHanded Random Walk" result of~\cite{AI} (motivated by the LeftHanded version of the LLL introduced by Pedgen~\cite{PegdenLLLL}), under which the flaw order $\pi$ can be chosen in a \emph{provably} beneficial way, yielding a ``responsibility" digraph. That is, both in the regenerative case and in the general case, one can use our $\gamma_i$ as the charge for each event in the responsibility digraph of~\cite{AI}.

\begin{theorem}[Main result]\label{lem:master}
If $\mathcal{A}$ results by applying a traceable flaw-choice mechanism on $(D,\rho,\theta)$ and there exist positive real numbers $\{\param_i\}$ such that for every flaw $f_i \in F$,
\begin{align}
\zeta_i := \frac{\gamma_i}{\param_i } \sum_{S \in \List(f_i)}  \prod_{j \in S} \param_j < 1 \enspace \label{genikoteri} , 
\end{align}
then a sink is reached within $(T_0+s)/\delta$ steps with probability at least $1-2^{-s}$, where $\displaystyle{\delta = 1 - \max_{i \in [m]} \zeta_i}$ and
\[
 T_0 		 = \log_2 \left( \max_{ \sigma \in \Omega } \frac{\theta(\sigma)  }{ \mu(\sigma) }\right) + \log_2 \left( \sum_{S \in \mathrm{Roots}(\theta) } \prod_{j \in S} \psi_j   \right)\enspace .
\]
\end{theorem}

\subsubsection{Proofs of Theorems~\ref{asymmetric} and~\ref{olala} from Theorem~\ref{lem:master}}

In Appendix~\ref{sec:forests}, we describe the \emph{Break Forests} and \emph{Recursive Forests}, into which we inject the witness sequences of the algorithms of Theorems~\ref{asymmetric} and~\ref{olala}, respectively. Theorem~\ref{asymmetric} follows from Theorem~\ref{lem:master} as Break Forests  satisfy the conditions of Theorem~\ref{lem:master} with $\mathrm{Roots}(\theta) = 2^{\Span(\theta)}$ and $\mathrm{List}(f) = 2^{\Gamma_R(f)}$. Theorem~\ref{olala} follows as Recursive Forests satisfy the conditions with $\mathrm{Roots}(\theta) = \Indep(\Span(\theta))$ and $\mathrm{List}(f) = \Indep(\Gamma_R(f))$.

\subsection{A Sharp Analysis and the Role of Flaw Choice}\label{sec:fc}

In Section~\ref{sec:misc_proofs} we prove that Lemma~\ref{lemma:regeneration} is tight for a rather large class of algorithms, including the algorithm of Moser-Tardos~\cite{MT} when each flaw fixes the values of its variables, as in SAT, and the algorithm of Harris and Srinivasan for permutations~\cite{SrinivasanPerm}.
\begin{theorem}\label{tight}
Let $\beta = \min_{\sigma \in \Omega} \mu(\sigma)>0$. If $(D,\rho)$ regenerate $\mu$ at every flaw and $D$ is atomic, then for every $W = w_1, w_2, \ldots, w_t \in \mathcal{W}_t$,
\begin{equation}
\beta \le 
\frac{\Pr[W_t = W]}{\prod_{i = 1}^{t} \mu(w_i)} \le \beta^{-1} \enspace .  \label{panw_katw}
\end{equation}
\end{theorem}
Equation~\eqref{panw_katw} tell us that an algorithm will  converge to a perfect object in polynomial time \emph{if and only if} the sum $\sum_{ W \in \mathcal{W}_t } \prod_{i =1}^{t} \mu(w_i) $ converges to a number less than $1$ as $t$ grows. In that sense, the quality of the algorithm's analysis depends solely on how well we approximate the set of possible witness sequences $\mathcal{W}_t$. 

The set $\mathcal{W}_t$ is clearly a function of how we choose which flaw to address in each step and therefore algorithmic performance clearly depends on the flaw choice mechanism (even more so, in this ``tight" case). However, in the Moser-Tardos analysis, as well as in the work of Harris and Srinivasan on permutations~\cite{SrinivasanPerm},  the flaw choice mechanism ``is swept under the rug''~\cite{mario_survey} and is allowed to be arbitrary. This can be explained as follows.   In those two settings, due to the symmetry of $\Omega$,  we can afford to approximate $\mathcal{W}_t$ in a way that completely ignores flaw choice, i.e., considering it adversarial, and still recover the LLL condition. In a very recent paper~\cite{Kolmogorov15}, Kolmogorov gives a more general symmetry condition under which the results of~\cite{AI} for the flaws/actions framework hold with arbitrary flaw choice. However, such symmetries can not be expected to hold in general settings, something reflected in Theorems~\ref{asymmetric} and~\ref{olala} in the specificity of the flaw-choice mechanism, while in Theorem~\ref{lem:master} it is reflected in the requirement of traceability.

\subsection{Applications: Incorporating Global Conditions}

To demonstrate the power of our framework we derive a novel bound for acyclic edge colorings, aimed at graphs of bounded degeneracy, a class including graphs of bounded treewidth. To get the result we heavily use the fact that we do not have to regenerate a measure (and so the result cannot be captured by the LLL). Unlike recent work on the problem~\cite{acyclic, kirousis} that also goes beyond the LLL, our result is established without any problem specific elements, but rather as a direct application of Theorem~\ref{olala}.

\section{Charging Flaws}\label{charging}

In Section~\ref{sec:setup} we defined how to assign to each flaw  $f_i$ a charge $\gamma_i$, depending on whether $(D,\rho)$ regenerate $\mu$ or not. We also stated that in the non-regenerative case $\gamma_i \ge \mu(f_i)$ always. 
Thus, ideally, we would like to use a sophisticated measure $\mu$ that assigns minimal probability mass to the flawed states and, at the same time, have $(D,\rho,\mu)$ that regenerate $\mu$ at every flaw. 
In reality, the more sophisticated $\mu$ is the harder regeneration becomes.  Therefore, realistically, we can either employ $(D,\rho, \mu)$ that regenerate $\mu$ at every flaw and get charges as small as possible, but for unsophisticated measures (like product measures), or we can forgo regeneration, pay the price $\gamma_i > \mu(f_i)$ to reflect the distortion of $\mu$ by $(D,\rho)$, and use more sophisticated measures. Crucially, making the latter choice typically also means that we can get a sparser causality graph by exploiting the flexibility afforded by not having to design $D$ so as to regenerate $\mu$, as in the non-regenerative case $D$ can be arbitrary.

\begin{lemma}\label{lem:comp_m_g}
$\gamma_i \ge \mu(f_i)$. 
\end{lemma}
\begin{proof}
By the definition of $\gamma_i$, in the regenerative case we have equality, while in the non-regenerative case,
\begin{equation}\label{eq:gamma_mu_comp}
\frac{\mu(f_i)}{\gamma_i}\le  \sum_{ \sigma \in f_i}  \frac{ \mu(\sigma) }{ b_i  \lambda_{i}^{\sigma}  }  = 
\sum_{ \sigma \in f_i} \sum_{ \tau \in A(i, \sigma)      } \frac{ \mu(\sigma) }{ b_i \lambda_{i}^{\sigma}  }   \rho_{i}( \sigma, \tau) \le
 \sum_{ \sigma \in f_i} \sum_{ \tau \in A(i, \sigma)      } \frac{\mu(\tau)}{b_i}  
 \le  1 \enspace ,
\end{equation}
where for the last inequality we used that $\sum_{ \sigma \in f_i}  \sum_{ \tau \in A(i, \sigma)}$ enumerates every $\tau \in \Omega$ at most $b_i$ times. 
\end{proof}

Observe that for any pair $(D,\mu)$ taking $\rho$ so that $(D,\rho,\mu)$ are harmonic, i.e., taking $\rho_i(\sigma,\tau) \propto \mu(\tau)$, minimizes $\lambda_i^{\sigma}$ for all $\sigma \in f_i$ \emph{simultaneously}. 
This optimality is the main reason that motivates harmonic algorithms. The other reason is the realization that designing $(D,\rho) $ that regenerate $\mu$ at every flaw is often achieved by $(D,\rho, \mu)$ being harmonic. As matter of fact, as we state in Theorem~\ref{atomic_oracles} below, if $D$ is atomic then $(D,\rho,\mu)$ being harmonic is \emph{necessary} for regeneration, a fact that also yields a characterization of the local structure of atomic digraphs regenerating a measure.

As a final remark, we note that  all results that correspond to ``algorithmizations of the LLL" correspond to the (very) special case where $(D,\rho)$ regenerate $\mu$ at every $f_i \in F$. 

\subsection{The atomic case}\label{sec:TheAtomicCase}
Atomic digraphs capture algorithms that appear in several settings, e.g., the Moser-Tardos algorithm~\cite{MT} when constraints are in CNF form, the algorithm of Harris and Srinivasan for permutations~\cite{SrinivasanPerm}, and others (see~\cite{AI}). Theorem~\ref{atomic_oracles} below asserts that $(D,\rho,\mu)$ being harmonic is a necessary condition for regeneration when $D$ is atomic. Observe that $(D,\rho,\mu)$ being harmonic means that we need not be concerned with the design of $\rho$ as it is implied by $(D,\mu)$. As for $D$ itself, the theorem implies that to build $A(i,\sigma)$ we must ``collect" arcs that satisfy~\eqref{synthiki_atomic} (while, presumably, keeping the causality graph as sparse as possible). These two facts offer guidance in designing $D$ so that $(D,\rho)$ regenerate $\mu$ at every flaw in atomic digraphs.

\begin{theorem} \label{atomic_oracles}
If $D$ is atomic and $(D,\rho)$ regenerate $\mu$ at every flaw, then $(D,\rho,\mu)$ are harmonic. Moreover, for every $i \in [m]$ and every $\sigma \in f_i$,
\begin{equation}
\sum_{\tau \in A(i,\sigma) } \mu(\tau) = \frac{\mu(\sigma)}{\mu(f_i)} \label{synthiki_atomic} \enspace .
\end{equation}
\end{theorem}
\begin{proof}
If $D$ is atomic, $\mu > 0$, and $(D,\rho)$ regenerate $\mu$ at every flaw $f_i$, it follows that for every $\tau \in \Omega$ there is \emph{exactly} one $\sigma \in f_i$ such that $\rho_i(\sigma,\tau)>0$. (And also that $\bigcup_{\sigma \in f_i} A(i,\sigma) = \Omega$). Therefore, regeneration at $f_i$ in this setting is equivalent to
\begin{equation}\label{eq:ptoma}
\text{for every $\sigma \in f_i$ and the unique $\tau \in A(i,\sigma)$:}\quad
 \rho_i(\sigma, \tau) = \mu(\tau)  \frac{\mu(f_i) }{ \mu(\sigma)}  \enspace .
\end{equation}
(Note that for given $D, \mu$ there may be no $\rho$ satisfying~\eqref{eq:ptoma}, as we also need that $\sum_{\tau \in A(i,\sigma)}\rho_i(\sigma,\tau) = 1$.)

Since $\rho_i(\sigma, \tau) \propto \mu(\tau)$ in~\eqref{eq:ptoma} we see that $\rho$ is harmonic. Summing~\eqref{eq:ptoma} over $\tau \in A(i,\sigma)$ yields~\eqref{synthiki_atomic}.
\end{proof}

\subsection{ Improved charges for the uniform measure case}\label{sec:AI_improvement}

As mentioned, the framework of~\cite{AI} amounts to the case where $D$ is atomic, $\mu$ is uniform, and $(D,\rho,\mu)$ is harmonic, so that in every step a uniformly random element of $A(i,\sigma)$ is selected. When $\mu$ is uniform, we prove in Section~\ref{sec:harmonic} that $\gamma_i$ is the inverse of $\min_{\sigma \in f_i} a_i^{\sigma}$, where $a_i^{\sigma} = |A(i,\sigma)|$ and, thus, our Theorems~\ref{asymmetric} and~\ref{olala} recover perfectly the results of~\cite{AI}. 

To deal with the case where $D$ is not naturally atomic, e.g., when a flaw occurs under  more than one value assignments to its variables, one can proceed to refine each $f_i \in F$ into $b_i = \max_{\sigma \in f_i} b_i^{\sigma}$ flaws so that $D$ becomes atomic. In subsection~\ref{uniform}, using the machinery developed for the general case in Section~\ref{sec:harmonic}, we remove the need to ``atomize" $D$ \emph{and} derive bounds dominating those that come from atomization. To achieve this we modify the walk so that each element $\tau \in A(i,\sigma)$ is selected with probability proportional to the inverse of its in-degree, i.e., $\rho_i(\sigma,\tau) = (b_i^{\tau} \sum_{\sigma' \in A(i,\sigma) } 1/b_i^{\sigma'} )^{-1}$. Doing this, we get that the charge we should assign to each flaw $f_i$ (assuming we are in the general case with uniform measure) is:

\[
\phi_i  
= \max_{(\sigma,\tau) \in D_i} \frac{b_i^{\tau}}{a_i^{\sigma}} \le \frac{b_i}{a_i}
\enspace ,
\]
where the right hand side above is the charge on flaw $f_i$ that would be assigned by atomization, potentially much greater than our bound $\max_{(\sigma,\tau) \in D_i} b_i^{\tau}/a_i^{\sigma}$.

\section{Proof of Lemma~\ref{lemma:regeneration}}

In Section~\ref{sec:probab}  we give the proof for  the regenerative case when the measure $\mu$ is sampleable (while $(D,\rho, \mu)$ are not necessarily harmonic). The proof for that case mimics that of~\cite{kirousis} and~\cite{HV}. In Section~\ref{sec:harmonic} we show the proof for the general (non-regenerative) case. The proof for the regenerative case when $\mu$ is not sampleable but   $D$ is atomic, is given as a special case of that proof in Section~\ref{sec:misc_proofs}. Finally, in Section~\ref{uniform} we show how to get the improved bounds for the general case with uniform measure, described in Section~\ref{sec:AI_improvement}.

\subsection{The Regenerative Case with Sampleable $\mu$}\label{sec:probab}
 
We start by proving Lemma~\ref{lemma:regeneration} for the case where $\theta = \mu$ and $(D,\rho)$ regenerate $\mu$ at every flaw.
 
\begin{lemma}\label{lem:prod_form}
If $\theta = \mu$ and $(D,\rho)$ regenerate $\mu$ at every $f_i \in F$, then for every sequence $W = w_1,\ldots,w_t$,
\[
\Pr[W_t = W] \le  \prod_{i = 1}^{t} \mu(w_i) 
\enspace .
\]
\end{lemma} 
\begin{proof}
To bound $\Pr[W_t = W]$ we will drop the requirement that flaw $w_i$ is selected by the flaw-choice mechanism at $\sigma_i$ and only require that $\sigma_i \in w_i$. To bound this latter probability we introduce a randomized process $C$ which given as input an arbitrary sequence of flaws $a_1, a_2, \ldots, a_t$ proceeds as follows: 
\begin{itemize}
\item
Select a state $\sigma_1$ according to $\theta$ and set $\mathsf{fail}(0)=0$
\item
For $i$ from 1 to $t$ do: 
	\begin{itemize}
	\item
	If $\sigma_i \not\in a_i$, then set $\mathsf{fail}(i)=1$ else set $\mathsf{fail}(i)=\mathsf{fail}(i-1)$
	\item
	Address $a_i$ at $\sigma_i$ according to $(D,\rho)$, i.e., set $\sigma_{i+1} = \tau \in A(i,\sigma)$ with probability $\rho_i(\sigma_i,\tau)$
	\end{itemize}
\end{itemize}
Observe that the sequence $\mathsf{fail}(i)$ is non-decreasing. By coupling $C$ with the algorithm we readily get
\begin{equation}\label{eq:cprocess}
\Pr\left[ W_t  = w_1,\ldots, w_t \right] \le 
\Pr_C\left[ \bigcap_{i=1}^t \sigma_i \in w_i \right] \enspace .
\end{equation}

For $t \ge 0$, let $P(t)$ be the proposition: for every $A = a_1,\ldots,a_t$ and every $\tau \in \Omega$, on input $A$
\begin{equation}\label{eq:regeneration}
\Pr_C\left[ \sigma_{t+1} = \tau \middle | \mathsf{fail}(t)=0 \right ]  = \mu(\tau) \enspace .
\end{equation}
We will prove that $P(t)$ holds for all $t \ge 0$ by induction on $t$. Along~\eqref{eq:cprocess} this readily implies the lemma. 

$P(0)$ follows from $\theta = \mu$. Assume that $P(t)$ holds for all $t <s$ and consider any sequence $a_1,\ldots,a_s$. If $\mathsf{fail}(s)=0$ then $\mathsf{fail}(s-1)=0$ as well which, by $P(s-1)$, implies that $\sigma_s$ is distributed according to $\mu$. Moreover, $\sigma_s \in a_{s}$ and, therefore, $\sigma_{s+1} $ is selected by addressing $a_s$ at $\sigma_s$. Therefore for every state $\tau \in \Omega$,
\begin{eqnarray*}
\Pr_C\left[ \sigma_{s+1} = \tau \middle | \mathsf{fail}(s)=0 \right ]  = 
\sum_{\sigma \in a_s} \frac{\mu(\sigma)}{\mu(a_s)} \rho_{[s]}(\sigma,\tau) = \mu(\tau) \enspace , 
\end{eqnarray*}
where the second equality holds because $(D,\rho)$ regenerate $\mu$ every flaw in $F$.
\end{proof}

\subsection{The General Case}\label{sec:harmonic}

Since the flaw addressed in each step depends only on the trajectory up to that point and not on any future randomness, the probability of any specific $t$-trajectory $\Sigma = \sigma_1 \xrightarrow{w_1}   \sigma_2 \xrightarrow{w_2}  \ldots    \sigma_{t} \xrightarrow{w_t}  \sigma_{t+1}$ is
\begin{equation}\label{eq:trivial_bound}
\theta(\sigma_1)
\prod_{i = 1}^{t} \rho_{[i]}(\sigma_i, \sigma_{i+1}) 
\enspace ,
\end{equation}
where recall that $[i] = j$ such that $w_i = f_j$. To bound $\Pr[W_t = W]$ we will sum~\eqref{eq:trivial_bound} over all $t$-trajectories with witness sequence $W$. 

Recall that $b_i^{\tau} = |\{\sigma \in f_i: \tau \in A(i,\sigma)\}|$. For each pair $\langle W, \sigma_{t+1}\rangle$, where $\sigma_{t+1} \in \Omega$, we construct an edge-weighted tree as follows. The root of the tree is $\sigma_{t+1}$. Let $P_{t} = \{\sigma: \sigma_{t+1} \in A([t],\sigma)\}$, i.e., $P_t$ contains the $b_{[t]}^{\sigma_{t+1}}$ states in $D$ with an arc to $\sigma_{t+1}$ labelled $f_{[t]}$ which, thus, are the possible states immediately prior to $\sigma_{t+1}$ in any trajectory with witness sequence $W$. The progeny of the root consists of a child for each $\sigma_t \in P_t$,  each parent-child edge weighted by $\rho_{[t]}(\sigma_t,\sigma_{t+1})$. Each child vertex acquires progeny in the same manner, i.e., it has one child per possible immediately prior state, the corresponding edge annotated by $\rho_{[t-1]}(\sigma_{t-1},\sigma_{t})$. And so on, until $W$ is exhausted. The constructed tree has the following properties:
\begin{itemize}
\item
Its root-to-leaf paths are in 1--1 correspondence with the trajectories compatible with $\langle W, \sigma_{t+1}\rangle$.
\item
The product of the numbers along each root-to-leaf path (trajectory) equals $\prod_{i = 1}^{t} \rho_{[i]}(\sigma_i, \sigma_{i+1})$.
\end{itemize}
Thus, to compute the probability of all trajectories with witness sequence $W$ and final state $\sigma_{t+1}$ it suffices to sum, over all root-to-leaf paths, the product of the probability of each path's leaf vertex under $\theta$ with the product of the weights along the path's edges. 

To bound this sum for a measure $\mu$ we define for every $i \in [m]$,
\begin{equation}\label{eq:res_def}
\pres_i^{\sigma} = \pres_i^{\sigma}(\mu)
= 
\mu(\sigma)
\max_{\tau \in A(i,\sigma)} \rho_i(\sigma,\tau) 
 \frac{b_i^{\tau}}{\mu(\tau)} \enspace .
\end{equation}
Observe that, by definition, for every $(\sigma,\tau)$ such that $\tau \in A(i,\sigma)$ we have $b_i^{\tau} \ge 1$ and, therefore, 
\begin{equation}\label{eq:rho_bound}
\rho_i(\sigma, \tau)  \le \frac{\res_{i}^{\sigma}}{b_i^{\tau}} \, \frac{\mu(\tau)}{\mu(\sigma)} \enspace .
\end{equation}
Therefore, we can bound the probability of any $t$-trajectory (root-to-leaf path) $\Sigma = \sigma_1,\ldots,\sigma_{t+1}$ by\begin{align}\label{eq:fotara}
\theta(\sigma_1)
\prod_{i = 1}^{t} \rho_{[i]}(\sigma_i, \sigma_{i+1}) 
\le 
\theta(\sigma_1)	\prod_{ i = 1}^{t} \frac{\res_{[i]}^{\sigma}}{b_{[i]}^{\sigma_{i+1}}} \frac{\mu(\sigma_{i+1})}{\mu(\sigma_i) } 
= \frac{\theta(\sigma_{1})}{\mu(\sigma_1)} 
\mu(\sigma_{t+1}) \prod_{i = 1}^{t} \frac{\res_{[i]}^{\sigma}}{b_{[i]}^{\sigma_{i+1}}}
\enspace .
\end{align}

If in the tree we now replace the weight $\rho_{[i]}(\sigma_i, \sigma_{i+1})$ of every edge by $\res_{[i]}^{\sigma}/b_{[i]}^{\sigma_{i+1}}$ and the weight of every leaf by $\xi\mu(\sigma_{t+1})$, where $\xi = \max_{\sigma \in \Omega}\{\theta(\sigma)/\mu(\sigma)\}$, we see by~\eqref{eq:fotara}  that the aforementioned sum over all root-to-leaf paths will give an upper bound on the total probability of trajectories with witness sequence $W$ and final state $\sigma_{t+1}$. A key thing to observe   is that after this edge-weight replacement, any vertex of the tree, say one corresponding to a state $\tau$, will have some number $b_{[i]}^{\tau}$ children, while each of its parent-child edges will have weight $\res_{[i]}^{\sigma}/b_{[i]}^{\tau}$, for some state $\sigma$. This fact puts as in a position to perform the summation. 

For each $i \in [m]$, let $\phi_i = \max_{\sigma \in f_i} \pres_{i}^{\sigma}$. Consider any vertex $v$ of the tree whose children are leaves and let $\tau$ be $v$'s state. Replacing the $b_{[1]}^{\tau}$ children of $v$ with a single leaf child, connected to $v$ by an edge of weight $\phi_{[1]}$, can only increase the contribution to the sum of the subtree rooted at $v$ since $b_{[i]}^{\tau} \times \res_{[i]}^{\sigma}/b_{[i]}^{\tau} = \res_{[i]}^{\sigma} \le \phi_i$. Proceeding to collapse the progeny of all other vertices in the same level of the tree as $v$ and then moving on to the next level, etc.\ collapses the entire tree to a single path whose product of edge-weights is $\prod_{i = 1}^{t} \phi_{[i]}$ and whose leaf was weight $\xi\mu(\sigma_{t+1})$, implying 
\begin{align*}
\Pr[W_t = W] \le \sum_{\sigma_{t+1} \in \Omega}
\xi\mu(\sigma_{t+1})
\prod_{i = 1}^{t} \phi_{[i]}
=
\xi  \prod_{i = 1}^{t} \phi_{[i]} \enspace .
\end{align*}
To conclude the proof for the general case observe that for every $i \in [m]$,
\begin{equation}\label{eq:itslate}
\phi_i = \max_{\sigma \in f_i} \pres_{i}^{\sigma} 
= 	\max_{\sigma \in f_i} \left\{\mu(\sigma) \max_{\tau \in A(i,\sigma)} \rho_i(\sigma,\tau) \frac{b_i^{\tau}}{\mu(\tau)} \right\}
\le	\max_{\tau \in \Omega} b_i^{\tau} \max_{\sigma \in f_i} \left\{\mu(\sigma) \max_{\tau \in A(i,\sigma)}  \frac{\rho_i(\sigma,\tau)}{\mu(\tau)}\right\}
= \gamma_i \enspace . 
\end{equation}

\subsection{The Atomic Case and Proof of Theorem~\ref{tight}}\label{sec:misc_proofs}

If $D$ is atomic and $(D,\rho)$ regenerate $\mu$ at every $f_i$, Theorem~\ref{atomic_oracles} implies that $(D,\rho,\mu)$ are harmonic and thus
\begin{equation}\label{eq:unify}
\xi_i^{\sigma} = \lambda_i^{\sigma} = \max_{\tau \in A(i,\sigma)} \left\{\rho_i(\sigma,\tau) \, \frac{\mu(\sigma)}{\mu(\tau)}\right\} = 
\max_{\sigma \in f_i} \frac{\mu(\sigma)}{\sum_{\sigma' \in A(i,\sigma)} \mu(\sigma')} = \mu(f_i) \enspace ,
\end{equation}
where the last equality follows from~\eqref{synthiki_atomic}. This establishes the regenerative case of Lemma~\ref{lemma:regeneration} for atomic $D$.

To prove Theorem~\ref{tight} we  note that Lemma~\ref{lemma:regeneration}, valid for any $(D,\rho,\mu,\theta)$, readily yields the upper bound. For the lower bound we observe that in order for $W \in \mathcal{W}_t$ there must exist at least one trajectory $\Sigma^*$ such that $W_t(\Sigma^*) =  W$. Since, by~\eqref{eq:unify}, we have $\lambda_i^{\sigma} = \rho_i(\sigma,\tau) \mu(\sigma)/\mu(\tau)$ we can conclude that
\[
\Pr[W_t = W ] \ge 
\Pr \left[\Sigma = \Sigma^* \right] = 
\theta(\sigma^*_1) \prod_{i = 1}^{t} \rho_{[i]}(\sigma^*_i, \sigma^*_{i+1}) =
\theta(\sigma^*_1) \prod_{i = 1}^{t} \lambda_{[i]}^{\sigma_i^*} \frac{\mu(\sigma^*_{i+1})}{\mu(\sigma^*_{i})} =
\mu(\sigma_{t+1} )  \prod_{i = 1 }^{t} \mu(w_i)   \enspace . 
\]

\subsection{The Special Case of the Uniform Measure}\label{uniform}

To get the improved bounds for the uniform measure recall that  $a_i = \min_{\sigma \in f_i} a_i^{\sigma} = \min_{\sigma \in f_i}  |A(i,\sigma)|$.  If $D_i$ is the subgraph of $D$ comprising the arcs labeled by $f_i$, then~\eqref{eq:itslate} yields
\[
\phi_i = \max_{\sigma \in f_i} \pres_{i}^{\sigma} 
= 	\max_{\sigma \in f_i} \max_{\tau \in A(i,\sigma)} \rho_i(\sigma,\tau) b_i^{\tau}
= \max_{(\sigma,\tau) \in D_i} \frac{b_i^{\tau}}{a_i^{\sigma}} \le \frac{b_i}{a_i}
\enspace.
\]

\section{ Proof of Theorem~\ref{lem:master}}\label{sec:bpproof}

Per the hypothesis of Theorem~\ref{lem:master}, each bad $t$-trajectory on $D$ is associated with a rooted labeled witness forest with $t$ vertices such that given the forest we can reconstruct the sequence of flaws addressed along the $t$-trajectory. Recall that neither the trees, nor the nodes inside each tree in the witness forest are ordered. To prove Theorem~\ref{lem:master} we will give $T_0 $ such that the probability that a $(T_0+s)$-trajectory on $D$ is bad is exponentially small in $s$. Let $\widetilde { \mathcal{W}_t } \supseteq  \mathcal{W}_t$ be the set of witness sequences of size $t$ that correspond to these forests (sometimes, when it is clear from the context, we abuse the notation and use $\widetilde {\mathcal{W}_t}$ to denote the set of of witness forests  themselves). Per our discussion above  (see Lemma~\ref{lemma:regeneration}) to prove the theorem it suffices to prove that  $\max_{ \sigma \in \Omega} \frac{\theta(\sigma) }{\mu(\sigma)}   \sum_{W \in \widetilde{ \mathcal{W}}_t} \prod_{i = 1}^{t} \gamma_{[i]} $ is exponentially small in $s$ for $t=T_0+s$. 

To facilitate counting we fix an arbitrary ordering $\pi$ of $F$ and map each  witness forest  into the unique ordered forest that results by ordering the trees in the forest according to the labels of their roots and similarly ordering the progeny of each vertex according to $\pi$ (recall that both the flaws labeling the roots and the flaws labeling the children of each vertex are distinct). 

Having induced this ordering for the purpose of counting, we will encode each witness forest as a rooted, ordered $d$-ary forest $T$ with exactly $t$ nodes, where $d = \max_{f \in F} |\List(f)|$. In a rooted, ordered $d$-ary forest both the roots and the at most $d$ children of each vertex are ordered. We think of the root of $T$ as having  reserved for each flaw $f \in \mathrm{Roots}(\theta)$ a slot. If $f \in \mathrm{Roots}(\theta)$ is the $i$-th largest flaw in $F$ according to $\psi$ then we fill the $i$-th slot (recall that the flaws labeling the roots of the witness forest are distinct and that, as a set, belong in the set $\mathrm{Roots}(\theta)$).

Each node $v$ of $T$ corresponds to a node of the witness forest and therefore to a flaw $f$ that was addressed at some point in the $t$-trajectory of the algorithm. Recall now that each node in the witness forest that is labelled by a flaw $f$  has children labelled by distinct flaws in $\List (f)$. We thus think of each node $v$ of $T$ as having precisely  $|\List(f)|$ slots reserved for each flaw $g \in   \List(f)$ (and, thus, at most $d$ reserved slots in total). For each  $g \in  \List(f)$ we fill the slot reserved for $g$ and make it a child of $v$ in $T$. Thus, from $T$ we can reconstruct the sequence of  flaws addressed with the algorithm.  To proceed, we use ideas  from~\cite{PegdenIndepen}. Specifically, we introduce a branching process that produces only ordered $d$-ary forests that correspond to witness forests and bound $\sum_{W \in \widetilde{\mathcal{W}_t}} \prod_{i = 1}^{t} \gamma_{[i]} $ by analyzing it.

Given any real numbers $0 < \param_i < \infty$ we define $x_i = \frac{\param_i}{\param_i +1} $ and write $\mathrm{Roots}(\theta) = \mathrm{Roots}$ to simplify notation. To start the process we produce the roots of the labeled forest by rejection sampling as follows: For each flaw $g \in F$ independently, with probability $x_g$ we add a root with label $g$. If the resulting set of roots is in $\mathrm{Roots}$ we accept the birth. If not, we delete the roots created and try again. In each subsequent round we follow a very similar procedure. Specifically, at each step, each node $u$ with label $\ell$ ``gives birth", again, by rejection sampling: For each flaw $g \in \List(\ell)$ independently, with probability $x_{g}$ we add a vertex with label $g$ as a child of $u$. If the resulting set of children of $u$ is in $\List(\ell)$ we accept the birth. If not, we delete the children created and try again. It is not hard to see that this process creates every possible witness forest  with positive probability. Specifically, for a vertex labeled by $\ell$, every set $S \not\in\List(\ell) $ receives probability 0, while every set $S\in \List(\ell)$ receives probability proportional to
\[
w_{\ell}(S) = \prod_{g \in S} x_g \prod_{h \in  \left(\List(\ell) \right) \setminus S} \left(1- x_h\right) \enspace .
\]

To express the exact probability received by each $S\in\List(\ell)$ we define
\begin{equation}\label{eq:d_def}
Q(S)  :=   \frac{ \prod_{g \in S} x_g }{\prod_{g \in S}(1 - x_g)   }  = \prod_{g\in S} \param_g 
\end{equation}
and let $Z_{\ell} = \prod_{f \in \left(  \List(\ell)  \right)}\left(1 - x_f\right)$. 
We claim that $w_{\ell}(S) = Q(S) \, Z_{\ell} $. To see the claim observe that
\[
\frac{ w_{\ell}(S)}
{Z_{\ell}}
=  \frac{ \prod_{g \in S} x_g \prod_{h \in  \left(\List(\ell) \right) \setminus S} \left(1- x_h\right) }
{\prod_{f \in  \List(\ell)  }(1 - x_f) }  
=  \frac{ \prod_{g \in S} x_g }{\prod_{g \in S}(1 - x_g)   }  
= Q(S) \enspace .
\]
Therefore, each $S\in\List(\ell)$ receives probability equal to
\begin{equation}\label{eq:sing_birth}
\frac{w_{\ell}(S)}{\sum_{B \in \List(\ell)} w_{\ell}(B)}
=
\frac
{Q(S) Z_{\ell}}
{\sum_{B \in\List(\ell)} Q(B) Z_{\ell}}
=\frac{Q(S)}{\sum_{B \in\List(\ell)} Q(B)}
 \enspace .
\end{equation}

Similarly, each set $ R \in \mathrm{Roots}$ receives probability equal to $Q(R) \left( \sum_{B \in \mathrm{Roots} }Q(B)  \right)^{-1}$.

\begin{lemma}\label{branchingLemma}
The branching process described above produces every tree $\phi \in \widetilde{\mathcal{W}}_t$ with probability 
\begin{align*}
p_{\phi} =  \left( \sum_{ S \in  \mathrm{Roots}}
\prod_{i \in S} \param_i  
 \right)^{-1}  \prod_{v \in \phi}  \frac{\param_v}{ \sum_{S \in\List(v)} Q(S)}
\end{align*}
\end{lemma}

\begin{proof} 
For each node $v$ of $\phi$ let $N(v)$ denote the set of labels of its children. By~\eqref{eq:sing_birth},
\begin{align*}
p_{\phi} & =  \frac{ Q(R) }{ \sum_{S \in \mathrm{Roots}} Q(S)  }  \prod_{v \in \phi} \frac{Q(N(v))}{\sum_{S \in\List(v)} Q(S)} \\
& =  \frac{ Q(R) }{ \sum_{S \in \mathrm{Roots}} Q(S)  } \cdot \frac{\prod_{v \in \phi \setminus R} \param_v} {\prod_{v \in \phi} \sum_{S \in\List(v)} Q(S)}\\
& = \left( \sum_{S \in \mathrm{Roots}} Q(S) \right)^{-1}   \prod_{v \in \phi} \frac{\param_v }{\sum_{S \in\List(v)} Q(S)}  \enspace .
\end{align*}
\end{proof}

Notice now that
\begin{eqnarray}
\sum_{W \in \widetilde{\mathcal{W}_t } } \prod_{i =1}^{ t} \gamma_{[i] } & = & \sum_{W \in \widetilde{\mathcal{W}_t }} \prod_{i =1}^{ t}   \frac{\zeta_{[i] }\, \param_{[i] }}{ \sum_{S \in \List({[i] })} Q(S) }  \nonumber \\
& \le & \left( \max_{ i \in F} \zeta_i  \right)^{t} \sum_{W \in \widetilde{\mathcal{W}_t }}  \prod_{i =1}^{ t}   \frac{\param_{[i] }}{ \sum_{S \in \List([i])} Q(S)}  \nonumber \\
& = & \left( \max_{ i \in F} \zeta_i  \right)^{t}  \sum_{W \in \widetilde{\mathcal{W}_t } }\left( p_{W}  \sum_{S \in \mathrm{Roots} }  Q(S) \right)   \nonumber \\
& = & \left( \max_{ i \in F} \zeta_i  \right)^{t}  \sum_{ S \in \mathrm{Roots} } Q(S) \label{no_kagkouro_versions}
\end{eqnarray}

Using~\eqref{no_kagkouro_versions} we see that the binary logarithm of the probability that the walk does not encounter a flawless state within $t$ steps is at most $t  \log_2 \left(  \max_{i \in F} \zeta_i \right) + T_0$, where 
\begin{eqnarray*}
 T_0 		& =& \log_2 \left( \max_{ \sigma \in \Omega } \frac{\theta(\sigma)  }{ \mu(\sigma) }\right) + \log_2 \left( \sum_{S \in \mathrm{Roots} } \prod_{i \in S} \param_i   \right)\enspace .
\end{eqnarray*}

Therefore, if $t = (T_0 + s) / \log_2 (1/ \max_{i \in F} \zeta_i) \le (T_0 + s) / \delta$, the probability that the  random walk  on $D$ does not reach a flawless state within $t$ steps is at most  $2^{-s}$.

\section{ Application to Acyclic Edge Coloring }\label{AECARA}

\subsection{Earlier Works and Statement of Result}

An edge-coloring of a graph is \emph{proper} if all edges incident to each vertex have distinct colors. A proper edge coloring is \emph{acyclic} if it has no bichromatic cycles, i.e., no cycle receives exactly two (alternating) colors. Acyclic Edge Coloring (AEC), was originally motivated by the work of Coleman et al.~\cite{coleman1,coleman2} on the efficient computation of Hessians. The smallest number of colors, $\chi'_a(G)$, for which a graph $G$ has an acyclic edge-coloring can also be used to bound other parameters, such as the oriented chromatic number~\cite{orient_col} and the star chromatic number~\cite{star_col}, both of which have many practical applications. The first general linear upper bound for $\chi'_a$ was given by Alon et al.~\cite{NogaLLL} who proved $\chi'_a(G) \le 64 \Delta(G)$, where $\Delta(G)$ denotes the maximum degree of $G$. This bound was improved to $16\Delta$ by Molloy and Reed~\cite{MRAEC} and then to $9.62(\Delta-1)$ by Ndreca et al.~\cite{Ndreca}. Attention to the problem was recently renewed due to the work of Esperet and Parreau~\cite{acyclic} who proved $\chi'_a(G) \le  4(\Delta-1)$, via an entropy compression argument, a technique that goes beyond what the LLL can give for the problem. Very recently, Giotis et al.\ improved the result of~\cite{acyclic} to $3.74 \Delta$.\smallskip

We give a bound of $(2+o(1))\Delta$ for graphs of bounded degeneracy. This not only covers a significant class of graphs, but demonstrates that our method can incorporate global graph properties. Recall that a graph $G$ is $d$-degenate if its vertices can be ordered so that every vertex has at most $d$ neighbors greater than itself. If $\mathcal{G}_d$ denotes the set of all $d$-degenerate graphs, then all planar graphs are in $\mathcal{G}_5$, while all graphs with treewidth or pathwidth at most $d$ are in $\mathcal{G}_d$. We prove the following.
\begin{theorem}\label{Aecaki}
Every $d$-degenerate graph of maximum degree $\Delta$ has an acyclic edge coloring with $\lceil (2+\epsilon)\Delta\rceil$ colors than can be found in polynomial time, where $\epsilon = 16\sqrt{d/\Delta}$.
\end{theorem}

\subsection{Background}

As will become clear shortly, the main difficulty in AEC comes from the short cycles of $G$, with 4-cycles being the toughest. This motivates the following definition.
\begin{definition} 
Given a graph $G=(V,E)$ and a, perhaps partial, edge-coloring of $G$, say that color $c$ is \emph{4-forbidden for $e \in E$} if assigning $c$ to $e$ would result in either a violation of proper-edge-coloration, or in a bichromatic 4-cycle containing $e$. Say that $c$ is 4-available if it is not 4-forbidden.
\end{definition}
Similarly to~\cite{acyclic,kirousis} we will use the following observation that the authors of~\cite{acyclic} attribute to Jakub Kozik. 
\begin{lemma}[\cite{acyclic}]\label{lem:2D}
In any proper edge-coloring of $G$ at most $2(\Delta-1)$ colors are 4-forbidden for any $e \in E$.
\end{lemma}
\begin{proof}
The 4-forbidden colors for $e = \{u,v\}$ can be enumerated as: (i) the colors on edges adjacent to $u$, and (ii) for each edge $e_v$ adjacent to $v$, either the color of $e_v$ (if no edge with that color is adjacent to $u$), or the color of some edge $e'$ which together with $e, e_v$ and an edge adjacent to $u$ form a cycle of length $4$. 
\end{proof}

Armed with Lemma~\ref{lem:2D}, the general idea is to use a palette $P$ of size $2(\Delta -1) + Q$ colors so that whenever we wish to (re)color an edge $e$ there will be at least $Q$ colors 4-available for $e$ (of course, assigning such a color to $e$ may cause one or more cycles of length at least 6 to become bichromatic).  At a high level, similarly to~\cite{kirousis}, our algorithm will be:
\begin{itemize}
\item
Start at a proper edge-coloring with no bichromatic 4-cycles. 
\item
While bichromatic cycles of length at least 6 exist, recolor the edges of one with 4-available colors.
\end{itemize}
Note that to find bichromatic cycles in a properly edge-colored graph we can just consider each of the $\binom{|P|}{2}$
pairs of distinct colors from $P$ and seek cycles in the subgraph of the correspondingly colored edges.

\subsection{Applying our Framework}

Given $G=(V,E)$ and a palette $P$ of $2(\Delta -1) + Q$ colors, let $\Omega$ be the set of all proper edge-colorings of $G$ with no monochromatic 4-cycle. Fix an arbitrary ordering $\pi$ of $E$ and an arbitrary ordering $\chi$ of $P$. For every even cycle $C$ of length at least 6 in $G$ fix (arbitrarily) two adjacent edges $e_1^C, e_2^C$ of $C$. \smallskip 

-- Our distribution of initial state $\theta$ assigns all its probability mass to the following $\sigma_1 \in \Omega$: color the edges of $E$ in $\pi$-order, assigning to each edge $e \in E$ the $\chi$-greatest 4-available color.

-- For every even cycle $C$ of length at least 6 we define the flaw $f_C = \{\sigma \in \Omega: C \text{ is bichromatic}\}$. Thus, a flawless $\sigma \in \Omega$ is an acyclic edge coloring of $G$.

-- The set of actions for addressing $f_C$ in state $\sigma$, i.e., $A(C, \sigma)$, comprises all $\tau \in \Omega$ that may result from the following procedure: uncolor all edges of  $C$ except for $e_1^{C},e_2^{C}$; go around $C$, starting with the uncolored edge that is adjacent to $e_2^C$, etc., assigning to each uncolored edge $e \in C$ one of the  4-available colors for $e$ at the time $e$ is considered. Thus, by lemma~\ref{lem:2D}, $|A(C,\sigma)| \ge Q^{|C|-2}$.
 
\begin{lemma}\label{atomicityclaim}
For every flaw $f_C$ and  state $\tau \in \Omega$, there  is  at most $1$ arc $\sigma \xrightarrow{C}  \tau$, i.e., 
$b_{C}^{\tau} \le 1$.
\end{lemma}
\begin{proof}
Given $\tau$ and $C$, to recover the previous state $\sigma$ it suffices to extend the bicoloring in $\tau$ of $e_1^C, e_2^C$ to the rest of $C$ (since $C$ was bichromatic in $\sigma$ and only edges in $C\setminus\{ e_1^{C},e_2^{C}  \}$ were recolored).
\end{proof}
Thus, taking $\mu$ to be uniform and $\rho$ such that $(D,\rho,\mu)$ is harmonic yields $\gamma_C = Q^{-|C|+2}$. \medskip

Let $R$ be the symmetric directed graph with one vertex per flaw where $f_C \rightleftarrows f_{C'}$ iff $C \cap C' \ne \emptyset$. Since a necessary condition for $f_C$ to potentially cause $f_{C'}$ is that $C \cap C' \ne \emptyset$, we see that $R$ is a supergraph of the causality digraph. Thus, if we run the {\sc Recursive Walk} algorithm with input $R$, to apply Theorem~\ref{olala} we need to evaluate for each flaw $f_C$ a sum over the subsets of $\Gamma_R(C)$ that are independent in $R$. To carry out this enumeration we observe that independence in $R$ implies edge-disjointness which, in turn, implies that in each (independent) set of cycles to be enumerated, no edge of $C$ appears in multiple cycles. Thus, to perform the enumeration it suffices to enumerate the subsets of edges of $C$ that appear in the cycles and for each appearing edge $e$ to enumerate all even cycles of length at least 6 containing $e$.   

Let $g(k) =   \max_{e \in E} |\{\text{$k$-cycles in $G$ that contain $e$}\}|$. If $\psi_C = \psi(|C|)$, then we can bound~\eqref{olala} as
\begin{eqnarray}
\frac{\gamma_C}{\psi_C} \sum_{S \in \Indep(\Gamma_R(C))}  \prod_{C' \in S} \psi_{C'} 
& \le & 
   \frac{ 1 }{  \psi_C Q^{|C|-2}}  
   \cdot  
   \sum_{ i = 0 }^{|C|}   
   \left( \binom{|C|}{i}  
   	\left(
    			\sum_{j = 3}^{ \infty } g(2j) \psi(2j)   
   \right)^{i}
   \right)  \nonumber \\
 & = &  \frac{ 1 }{\psi_C Q^{|C|-2}}  \cdot  \left ( 1 +   \left(  \sum_{j = 3}^{ \infty }   g\left(2j \right) \psi \left(2j \right)  \right)  \right)^{|C|}  \enspace . \label{auti}
\end{eqnarray}

We will prove the following structural lemma relating degeneracy to $g$.
\begin{lemma}\label{kral}
If $G \in \mathcal{G}_d$ has maximum degree $\Delta$, then $g(k) \le 2 (4 d \Delta)^{(k-2)/2}$.
\end{lemma}
We are thus left to choose $\psi$ such that for every even $|C| \ge 6$, the r.h.s.\ of~\eqref{auti} is strictly less than 1. 
Taking $\psi(k) = \left(8d\Delta\right)^{-\frac{k-2}{2}}$ and $Q = \lceil 16\sqrt{d \Delta} \rceil$ we see that for all $k \ge 6$,
\[
\frac{ (8d\Delta)^{\frac{k-2}{2}} }{\left(16\sqrt{d\Delta}\right)^{k-2}}  \cdot  \left ( 1 + 2  \left(  \sum_{j = 3}^{ \infty }   (4d\Delta)^{j-1}   \left(8d\Delta\right)^{-j+1} \right)  \right)^{k} =
2^{-\frac{3k}{2}+5} < 1 \enspace .
\]
Regarding the running time, notice that $\delta \ge 1-2^{-4} = 15/16$ and that it can easily be seen that that $T_0$ is a polynomial in $|E|$, $\Delta$ and the number of colors used. 

\begin{proof}[Proof of Lemma~\ref{kral}]
Fix any edge $e=\{u,v\} \in E$. To enumerate the $k$-cycles containing $e$ we will partition them into equivalence classes as follows. First we orient all edges of $G$ arbitrarily to get a digraph $D$. Consider now the two possible traversals of the path $C \setminus \{u,v\}$, i.e., the one starting at $u$ and the one starting at $v$. For each traversal generate a string in $\{0,1\}^{k-2}$ whose characters correspond to successive vertices of the path, other than the endpoints, and denote whether the corresponding vertex was entered along an edge oriented in agreement (1) or in disagreement (0) with the direction of travel. Observe that each of the $k-3$ edges of $C$ that have no vertex from $\{u,v\}$ will create a 1 in one string and a 0 in the other. Therefore, at least one of the strings will have at least $\lceil (k-3)/2\rceil = (k-2)/2$ ones. Select that string, breaking ties in favor of the string corresponding to starting at $u$. Finally, prepend a single bit to the string to designate whether the winning string corresponded to $u$ or to $v$. The string denotes $C$'s equivalence class.

To enumerate all $k$-cycles containing $e$ we can thus enumerate all binary strings of length $k-1$ and use each string to select the $k-2$ other vertices of the cycle as follows: after reading the first character to decide whether to start at $u$ or at $v$, we interpret each successive character to indicate whether we should choose among the out-neighbors or the in-neighbors of the current vertex. By the string's construction, we will chose among out-neighbors $q \ge (k-2)/2$ times. If $\mathrm{Out}$ and $\mathrm{In}$ are upper bounds on the out- and in-degree of $D$, respectively, then the total number of cycles per string (class) is bounded by $\mathrm{Out}^q \mathrm{In}^{k-2-q}$.

To conclude the argument we note that since $G \in \mathcal{G}_d$ we can direct its edges so that every vertex has out-degree at most $d$ by repeatedly removing any vertex $v$ of current degree at most $d$ (it always exists) and, at the time of removal, orienting its current neighbors away from $v$.
\end{proof}

\section*{Acknowledgements}
We are grateful to Dan Kral for providing us with Lemma~\ref{kral} and to Louis Esperet for pointing out an error in our application of Theorem~\ref{olala} to yield Theorem~\ref{Aecaki} in a previous version of the paper. FI is thankful to Alistair Sinclair for many fruitful conversations.

\bibliographystyle{plain}
\bibliography{smoser}

\newpage
\appendix

\newpage

\section{Mapping Bad Trajectories to Forests}\label{sec:forests}

In this section, we show how to represent each sequence of $t$ steps that does not reach a sink as a forest with $t$ vertices, where the forests have different characteristics for each of the walks of Theorems~\ref{asymmetric}, \ref{olala}.

\subsection{Forests of the Permutation Walk (Theorem~\ref{asymmetric})} \label{TheProof}

For $S \subseteq F$,  we denote by $I_{\pi}(S) = I(S)$  the greatest element of $S$ according to $\pi$.  We will sometimes write $I(\sigma)$ to denote $I\left(  U(\sigma) \right)$. We first show how to represent the witness sequences of the Permutation Walk as sequences of sets.

Let $B_i$ be the set of flaws ``introduced" by the $i$-th step of the walk, where a flaw $f_j$ is said to ``introduce itself" if it remains present after an action from $A(j, \cdot)$ is taken. Formally,
\begin{definition}
Let  $B_0 = U(\sigma_1)$. For $1 \le i\le t-1$, let $B_i = U(\sigma_{i+1}) \setminus ( U (\sigma_i) \setminus I(\sigma_i) )$. 
\end{definition}

Let $B_i^* \subseteq B_i$ comprise those flaws addressed in the course of the trajectory. Thus, $B_i^* = B_i \setminus \{O_i \cup N_i \}$, where $O_i$ comprises any flaws in $B_i$ that were eradicated ``collaterally" by an action taken to address some other flaw, and $N_i$ comprises any flaws in $B_i$ that remained present in every subsequent state after their introduction without being addressed. Formally,
\begin{definition}\label{def:bs}
The \emph{Break Sequence} of a $t$-trajectory is $B_0^*, B_1^*, \ldots, B_{t-1}^*$, where for $0 \le i \le t-1$,
\begin{align*}
O_i 		& 	= 	\{f \in B_i \mid  \exists j \in [i+1,t] :   f \notin U(\sigma_{j+1})  \wedge  \forall \ell \in [i+1,j]:   f \ne w_{\ell} \} \\
N_i 		&  	= 	\{f \in B_i \mid \forall j \in [i+1, t] :    f \in 	 U(\sigma_{j+1})  \wedge  \forall \ell \in [i+1,t]:   f \ne w_{\ell} \} \\
B_i^*  	&  = B_i \setminus \{O_i \cup N_i \} \enspace .
\end{align*}

\end{definition}

Given $B_0^*,B_1^*,\ldots,B_{i-1}^*$ we can determine the sequence of flaws addressed $w_1, w_2, \ldots, w_i$ inductively, as follows. Define $E_1 = B_0^*$, while for $i \ge 1$,
\begin{equation}\label{eq:ri}
E_{i+1} = (E_{i} - w_i) \cup B_i^*  \enspace .
\end{equation}
By construction, the set $E_i \subseteq U(\sigma_i)$ is guaranteed to contain $w_i = I(\sigma_i)=I(U(\sigma_i))$. Since $I = I_{\pi}$ returns the greatest flaw in its input according to $\pi$, it must be that $I_{\pi}(E_i) = w_i$. We note that this is the only place we ever make use of the fact that the function $I$ is derived by an ordering of the flaws, thus guaranteeing that for every $f \in F$ and $S \subseteq F$, if $I(S) \neq f$ then $I(S\setminus f) = I(S)$. 

We now give a 1-to-1 map, from Break Sequences to vertex-labeled unordered rooted forests. Specifically, the \emph{Break Forest} of a bad $t$-trajectory $\Sigma$ has $|B_0^*|$ trees and $t$ vertices, each vertex labelled by an element of $W(\Sigma)$. To construct it we first lay down $|B_0^*|$ vertices as roots and then process the sets $B_1^*, B_2^*, \ldots$ in order, each set becoming the progeny of an already existing vertex (empty sets, thus, giving rise to leaves). 
\begin{algorithm}\caption*{{\bf Break Forest Construction}}
\begin{algorithmic}[1]
\State Lay down $|B_0^*|$ vertices, each labelled by a different element of $B_0^*$, and let $V$ consist of these vertices
\For {$i=1$ to $t-1$} 
\State Let $v_i$ be the vertex in $V_i$ with greatest label according to $\pi$
\State Add $|B_i^*|$ children to $v_i$, each labelled by a different element of $B_i^*$
\State Remove $v_i$ from $V$; add to $V$ the children of $v_i$.
\EndFor
\end{algorithmic}
\end{algorithm}

Observe that even though neither the trees, nor the nodes inside each tree of the Break Forest are ordered, we can still reconstruct $W(\Sigma)$ since the set of labels of the vertices in $V_i$ equals $E_i$ for all $0 \le i\le t-1$.

\subsection{Forests of the Recursive Walk (Theorem~\ref{olala})}\label{rec_forests}

We will represent each witness sequence $W =  W(\Sigma)$ of the Recursive Walk as a vertex-labeled unordered rooted forest, having one tree per invocation of procedure {\sc{address}} by procedure {\sc{eliminate}}. Specifically, to construct the \emph{Recursive Forest} $\phi=\phi(\Sigma)$ we add a root vertex per invocation of {\sc{address}} by {\sc{eliminate}} and one child to every vertex for each (recursive) invocation of {\sc{address}} that it makes. As each vertex corresponds to an invocation of {\sc{address}} (step of the walk) it is labeled by the invocation's flaw-argument. Observe now that (the invocations of {\sc{address}} corresponding to) both the roots of the trees and the children of each vertex appear in $W$ in their order according to $\pi$. Thus, given the unordered rooted forest $\phi(\Sigma)$ we can order its trees and the progeny of each vertex according to $\pi$ and recover $W$ as the sequence of vertex labels in the preorder traversal of the resulting ordered rooted forest.

Recall the definition of graph $G$ on $F$ from Definition~\ref{defn:G}. We will prove that the flaws labeling the roots of a Recursive Forest are independent in $G$ and that the same is true for the flaws labelling the progeny of every vertex of the forest. To do this we first prove the following.

\begin{proposition}\label{prop:rid}
If {\sc{address}}($i,\sigma$) returns at state $\tau$, then $U(\tau) \subseteq U(\sigma) \setminus (\Gamma_R(f_i) \cup \{f_i\})$. 
\end{proposition}
\begin{proof}
Let $\sigma'$ be any state subsequent to the {\sc{address}}($i,\sigma$) invocation. If any flaw in $U(\sigma) \cap \Gamma_R(f_i)$ is present at $\sigma'$, the ``while" condition in line~\ref{a_key_diff} of the Recursive Walk prevents {\sc{address}}($i,\sigma$) from returning. On the other hand, if $f_h \in \Gamma_R(f_i) \setminus U(\sigma)$ is present in $\sigma'$, then there must have existed an invocation {\sc{address}}($j, \sigma''$), subsequent to invocation {\sc{address}}($i, \sigma$), wherein addressing $f_j$ caused $f_h$. Consider the last such invocation. If $\sigma'''$ is the state when this invocation returns, then $f_h \not\in U(\sigma''')$, for otherwise the invocation could not have returned, and by the choice of invocation, $f_h$ is not present in any subsequent state between $\sigma'''$ and $\tau$.  
\end{proof}

Let $([i],\sigma^i)$ denote the argument of the $i$-th invocation of {\sc{address}} by {\sc{eliminate}}. By Proposition~\ref{prop:rid}, $\{U(\sigma^i)\}_{i \geq 1}$ is a decreasing sequence of sets. Thus, the claim regarding the root labels follows trivially: for each $i \ge 1$, the flaws in $\Gamma_R(f_i) \cup f_i$ are not present in $\sigma^{i+1}$ and, therefore, are not present in $U(\sigma^j)$, for any $j \ge i+1$. The proof for the children of each node is essentially identical. If a node corresponding to an invocation of {\sc{address}} has $q$ children, corresponding to $q$ (recursive) invocations with arguments $\{(a_i,\tau^i)\}_{i=1}^q$, then the sequence of sets $\{U(\tau^i)\}_{i = 1}^q$ is decreasing. Thus, the flaws in $\Gamma_R(a_i) \cup \{a_i \}$ are not present in $\tau^{i+1}$ and, therefore, not present in $U(\tau^j)$, for any $j \ge i+1$. 

\end{document}